\documentclass[12pt]{article}
\usepackage{tikz}
\usepackage{enumitem}
\usepackage{latexsym,verbatim}
\usepackage{url,color}
\usepackage{amsmath,amssymb,amsthm,pict2e}

\newcommand{\F}{\mathbb{F}}
\newcommand{\x}{\mathbf{x}}
\newcommand{\A}{{\mathcal{A}}}
\newcommand{\U}{{\mathcal{U}}}
\newcommand{\LL}{{\mathcal{L}}}

\newcommand{\vertex}{{\mathcal{V}}}
\newcommand{\edge}{{\mathcal{E}}}
\newcommand{\Num}{{\mathcal{N}}}
\newcommand{\graph}{{\mathcal{G}}}
\newcommand{\SV}{{\mathcal{SV}}}

\newtheorem{theorem}{Theorem}[section]
\newtheorem{lem}[theorem]{Lemma}

\theoremstyle{definition}
\newtheorem{example}[theorem]{Example}
\newtheorem{question}{Question}
\newtheorem{definition}[theorem]{Definition}
\newtheorem{const}[theorem]{Construction}
\newtheorem{rem}[theorem]{Remark}

\newenvironment{eg}{\begin{example} \upshape}{\hfill $\Box$ \end{example}}
\newenvironment{lemma}{\begin{lem} \em}{\hfill $\Box$ \end{lem}}
\newenvironment{defn}{\begin{definition} }{\end{definition}}
\newenvironment{thm}{\begin{theorem} \upshape}{\hfill $\Box$ \end{theorem}}

\begin{document}

\title{Functional repair codes: a view from projective geometry}
\author{Siaw-Lynn Ng\footnote{ Information Security Group, Royal Holloway, University of London, Egham, Surrey, TW20 0EX,
U.K. S.Ng@rhul.ac.uk.}, Maura B. Paterson\footnote{Department of Economics, Mathematics and Statistics, Birkbeck, University of London, Malet St, London WC1E 7HX, U.K. m.paterson@bbk.ac.uk.}}
\maketitle


\begin{abstract}
Storage codes are used to ensure reliable storage of data in
distributed systems.  Here we consider functional repair codes, where
individual storage nodes that fail may be repaired efficiently
and the ability to recover original data and to further repair failed
nodes is preserved.  There are two predominant approaches to repair
codes: a coding theoretic approach and a vector space approach.  We
explore the relationship between the two and frame the later in terms
of projective geometry.  We find that many of the constructions proposed in the literature can be seen to arise from natural and well-studied geometric objects, and that this perspective gives a framework that provides opportunities for generalisations and new constructions that can lead to greater flexibility in trade-offs between various desirable properties. We also frame the cut-set bound obtained from network coding in terms of projective geometry.

We explore the notion of {\em strictly functional} repair codes, for which there exist nodes that \emph{cannot} be
replaced exactly.  Currently only one known example is given in the literature, due to Hollmann and Poh. We examine this phenomenon from a projective geometry point of view, and discuss how strict functionality can arise.  

Finally, we consider the issue that the view of a repair code as a collection of sets of vector/projective
subspaces is recursive in nature and makes it hard to visualise what a
collection of nodes looks like and how one might approach a
construction. Here we provide another view of using directed graphs that gives us non-recursive criteria for determining whether a family of collections of subspaces constitutes a function, exact, or strictly functional repair code, which may be of use in searching for new codes with desirable properties.
\end{abstract}

\section{Introduction}

The growth of data and an increasing reliance on digital information have
led to much research into ensuring that data can be stored reliably.
One predominant solution is the use of storage codes for distributed
storage systems: a database is coded and stored in multiple nodes
(servers) in such a way that if a number of nodes fail, the data can
still be recovered from the functioning nodes.  One technique used in
practice (for example, RAID \cite{RAID}, Total Recall
\cite{TotalRecall}) is that of erasure coding: for instance, MDS codes
such as the Reed-Solomon code (\cite{MacWilliamsSloane})
can be used to ensure that any number of node failures up to
a certain threshold does not impede the \emph{recovery} of the entire
database.  However, many distributed storage systems also require
additional resilience properties.  In particular, we may want to mitigate node failures: if a node should fail, we would like to
\emph{repair} it using information in some of the functioning nodes so
that the the recovery property of the system still holds.  Clearly one
could do that by simply recovering the entire database and re-encoding
it.  This involves a sometimes unacceptable overhead in storage and
communication.  Much work has been done to minimise the amount of data
to be stored and the amount of data to be transmitted for
repair. Using techniques from network coding, Dimakis \textit{et al.} 
(\cite{NetworkCoding}) showed that one could significantly reduce
the amount of data to be communicated for repair and showed that there
is a tradeoff between storage and repair efficiency.  Since then much
work has been done on modelling and constructing efficient repair
codes. Here we consider two strands of this work.

In \cite{productmatrix}, Rashmi \textit{et al.} proposed a
product-matrix framework for repair codes. This is an essentially
coding theoretic approach, where the database is treated as messages
that are encoded using a generator matrix.  The resulting codewords
are then stored in individual nodes.  Using this framework, repair
codes can be constructed with parameters that sit on various points on
the storage-repair tradeoff curve.  On the other hand, in
\cite{HollmannPoh}, Hollmann and Poh viewed a repair code as a
collection of sets of subspaces of a vector space. Recovery
corresponds to generating the vector space while repair corresponds to
generating a subspace.  In this paper (Section \ref{sub:motivation})
we explore the relationship between these two models and motivate the
interpretation of the vector space model in terms of projective
geometry.  We will see that many constructions arise naturally from
looking at repair codes from a projective geometric point of view
(Section \ref{sec:constructions}) and these include the constructions
in \cite{HollmannPoh, subspacecode}.  We also frame a special case of
the cut-set bound that Dimakis \textit{et al.} obtained from network
coding technique (\cite{NetworkCoding}) in terms of projective spaces.
This proves to be relatively straight forward compared to the original
proof.

There are broadly speaking two types of repair.  In {\em exact}
repair, if a node fails then the new node constructs exactly the same
symbols that the failed node stored.  In {\em functional} repair, the
new node does not necessarily contain the same symbols as the failed
node, but the set of nodes after repair should remain a repair code:
one should still be able to recover the original database, and future
repair should be possible.  We will call a functional repair code that
does not admit exact repair a \emph{strictly functional} repair code.
In this paper we will also clarify what exact and functional repair
means.  One interesting question is whether there exists strictly functional
repair codes.  In Section \ref{sec:constructions} we see that there
are repair codes that can be both functional and exact, but in
\cite{HollmannPoh} there is a construction that is strictly
functional.  This appears to be the only example in the literature so far.
Another aim of this paper (Sections \ref{sub:focal} and \ref{sec:strictly})
is to examine this structure from a
projective geometry point of view and to see if it can be generalised.
We give another example of a strictly functional repair code
which arises from a familiar structure in projective planes 
(Section \ref{sub:dualarcs}).

The study of the strictly functional construction from
\cite{HollmannPoh} is also motivated by the following: the view of a
repair code as a collection of sets of vector/projective subspaces is
recursive in nature: one must be able to derive a new subspace from an
``admissible'' set, and the new subspace, together with all but one of
the subspaces from the ``admissible'' set must again be
``admissible''.  This models the repair property, insisting that
future repairs must be possible.  However, this recursive nature makes
it hard to visualise what a collection of admissible sets look like:
it is hard to discern the ``global view'' of the whole set of nodes
from the ``local view'' of individual node repairs. In the
construction of this strictly functional repair code, a description is
given that uses symmetry to bypass the recursiveness of the definition. This
naturally leads to the question of whether this can be generalised,
and also motivates another view using directed graphs.  We discuss
exact and functional repairs in terms of the properties of these
graphs in Section \ref{sec:digraphs}.

We will make these aims more precise when we introduce notation.  We
would like to note that constructing new efficient
storage codes is not the primary focus of this work, even though many objects in projective geometry appears
to offer good repair as well as flexibility in terms of resilience and
trade-offs between locality and repair. We intend rather to clarify
the definition and properties of functional repair codes, and to
consider their possible relationship with other combinatorial objects.

The structure of this paper is as follows: we will give definitions
and introduce notation in Section \ref{sec:general}, and consider
some motivation for phrasing things in terms of projective geometry
(Section \ref{sub:motivation}).  In Section \ref{sec:constructions} we examine
functional repair codes arising from projective geometry objects, and
study in further detail the strictly functional repair code of
\cite{HollmannPoh} in Sections \ref{sec:strictly} and \ref{sec:simpler}. 
In Section \ref{sec:digraphs} we consider functional repair codes as
digraphs, and in Section \ref{sec:further} we discuss further work.

\section{Definitions and basic properties} \label{sec:general}
An $(m; n, k,r,\alpha,\beta)$-functional repair code stores $m$
information symbols from some finite alphabet $\F$, encoded across $n$
storage nodes.  Each storage node can hold $\alpha$ symbols.  The following
properties hold:
\begin{enumerate}[label=(\Roman*)] 
\item \label{recover} (Recovery)

The original information can be recovered from the data stored on $k$
nodes (the recovery set).

\item \label{repair} (Repair)

If a storage node fails then a newcomer node contacts some set of $r$
surviving nodes (the repair set) and downloads $\beta$ symbols from
each of these $r$ nodes. From these symbols the newcomer node
constructs and stores $\alpha$ symbols in such a way that
\ref{recover} holds and \ref{repair} holds if another node fails.
\end{enumerate}

We note that there is a dichotomy in the definition of the repair set:
in some work (for example, \cite{NetworkCoding,productmatrix}) it is
stipulated that the repair set is \emph{any} set of $r$ surviving
nodes, while in others (for example, \cite{HollmannPoh, Silberstein})
it is only required that there exists \emph{some} $r$ nodes to form
the repair set.  Similarly for the recovery set.  We will continue this
discussion after Definition \ref{defn:functional}.

In {\em exact}
repair, if a node fails then the newcomer node constructs exactly the
same symbols that the failed node stored, while in {\em functional} repair
the newcomer node does not necessarily contain the same symbols as the
failed node, but the set of $n$ nodes after repair should remain an
$(m; n,k,r,\alpha,\beta)$-functional repair code.
A functional repair code that does not admit
exact repair is a \emph{strictly functional} repair code.  (We will make
these definitions more precise in what follows.)  The focus of this paper is
functional repair.

In \cite{HollmannPoh} and various subsequent work, a functional repair code is viewed as a
collection of sets of subspaces of an $m$-dimensional vector space over a finite field $\F_q$.  The underlying storage codes work as follows:
\begin{itemize}
\item For $i$ with $0\leq i\leq n-1$, the $i^{\rm th}$ node is assigned a vector space represented by a specified basis $\{\mathbf{v}^i_0,\mathbf{v}^i_,\dotsc,\mathbf{v}^i_{\alpha-1}\}$.
\item To store a message $\mathbf{x}=(x_0,x_1,\dotsc,x_{m-1})\in \F_q^m$, each node $i$ with $0\leq i\leq n-1$ stores the $\alpha$ scalar values $\{\mathbf{x}\cdot\mathbf{v}^i_0,\mathbf{x}\cdot\mathbf{v}^i_,\dotsc,\mathbf{x}\cdot\mathbf{v}^i_{\alpha-1}\}$.
\item If the vector $(1,0,\dotsc,0)$ is in the span of a set of vectors $\{\mathbf{u_0},\mathbf{\dotsc,u_{t-1}}\}$, then the values $\{\mathbf{x}\cdot\mathbf{u_0},\dotsc,\mathbf{x}\cdot\mathbf{u_{t-1}}\}$ can be used to recover $x_0$.  For, if $(1,0,\dotsc,0)=\sum_{i=0}^{t-1}a_i\mathbf{u_i}$ for $a_i\in\F_q$, then $x_0=\sum_{i=0}^{t-1}a_i\mathbf{x}\cdot\mathbf{u_i}$.  If the vectors $\{\mathbf{u_0},\mathbf{\dotsc,u_{t-1}}\}$ span $\F_q^m$ then the entire message $\mathbf{x}$ can similarly be recovered from these values.
\end{itemize}

The properties of the storage code are hence determined by the relationship between the subspaces that correspond to the nodes, in particular, the spans and the intersections of these subspaces.  The projective space ${\rm PG}(m-1,q)$ provides a very natural setting for studying spans and intersections in $\F_q^m$.  It can make the relationship between spaces easier to visualise and, furthermore, many natural geometric structures in ${\rm PG}(m-1,q)$ have well-understood span/intersection properties that can be useful in constructing storage codes.  In what follows we will translate the vector-space definitions of \cite[Definitions 3.1, 3.2]{HollmannPoh} into the language of projective spaces.  We will see that this provides new insight into existing constructions of repair codes, such as \cite{HollmannPoh, subspacecode},  as well as suggesting useful frameworks for new construction of such codes.  

\begin{defn}[$(r, \beta)$-repair] \label{defn:repair}
Let $\Sigma={\rm PG}(m-1,q)$ be an $(m-1)$-dimensional projective space
over the finite field $\F_q$.  
We say that we can obtain a subspace $U'$ of $\Sigma$ 
from a set $\U$ of subspaces of $\Sigma$
by $(r, \beta)$-repair if there is an
$r$-subset $\{U_{i_1}, \ldots, U_{i_r}\}$ in $\U$ such that there
exists a $(\beta-1)$-dimensional subspace $W_{i_j} \subseteq U_{i_j}$
for each $i_j$ such that 
$U' \subseteq \langle W_{i_1}, \ldots, W_{i_r} \rangle$.
\end{defn}

\begin{defn}[Functional repair codes]\label{defn:functional}
Let $\Sigma={\rm PG}(m-1,q)$ and let $\A$ be a collection of $(n-1)$-sets 
$\U$ of $(\alpha-1)$-dimensional subspaces of $\Sigma$ such that:

\begin{enumerate}[label=(\Alph*)]
\item \label{pjfn:recovery} (Recovery)

For each set $\U\in\A$ there is a $k$-subset $\{U_{i_1},\dotsc,U_{i_k}\}$ of the subspaces in $\U$ whose span is all of $\Sigma$.

\item \label{pjfn:repair} (Repair)

Given any $(n-1)$-set $\U = \{U_1, \ldots, U_{n-1}\}$ in $\A$, there
exists an $(\alpha-1)$-dimensional subspace $U_n\subset \Sigma$ that can be obtained
from $\U$ by $(r, \beta)$-repair, such that for every $i=1, \ldots
n-1$, $\U \cup \{U_n\} \setminus \{U_i\}$ is again in $\A$.
\end{enumerate}
We will call $(\Sigma={\rm PG}(m-1, q), \A)$ an $(m; n,
k,r,\alpha,\beta)$-functional repair code (or $(m; n,
k,r,\alpha,\beta)$-FRC for convenience).  
\end{defn}

Here $\A$ corresponds to all possible sets of $n-1$ subspaces that
belong to the nodes remaining after a single node has failed.  The
repair property ensures that there is always a suitable subspace that
can be constructed by $(r,\beta)$-repair from these nodes in order to
construct a replacement for the node that has failed.  Clearly here we
require that there exists \emph{some} recovery set and \emph{some}
repair set, although in many of the constructions we describe in
Section \ref{sec:constructions}, repair and recovery can be effected
by arbitrary sets.  We will clarify each case as we go along.

To avoid triviality we assume that 
$m, n \ge 2$, 
$1 \le k < n$, 
$k \le r \le n-1$, 
$1 \le \alpha \le m-1$, 
and $1 \le \beta \le \alpha$.

\begin{definition}
Let $(\Sigma={\rm PG}(m-1, q), \A)$ be an $(m; n, k,r,\alpha,\beta)$-functional
repair code.  An $n$-set $\{U_1, \ldots, U_n\}$ of
$(\alpha-1)$-dimensional subspaces of $\Sigma$ with the property that $\{U_1,
\ldots, U_n\} \setminus \{U_j\} \in \A$ for all $j \in \{1, \ldots,
n\}$ is said to be {\em repairable}.
\end{definition}
It is the repairable sets corresponding to $(\Sigma,\A)$ that can be used as storage codes; if any node fails, the repair property then ensures that the resulting $(n-1)$-set permits a new repairable set to be obtained through $(r,\beta)$-repair.  Now we define exact and strictly functional repairs:

\begin{defn}[Exact repair]\label{defn:exact}
Let $(\Sigma={\rm PG}(m-1, q), \A)$ be an $(m; n, k,r,\alpha,\beta)$-functional
repair code.  We say that $(\Sigma, \A)$ is an exact repair
code if for any repairable set $\{U_1, \ldots, U_n\}$ we have the additional property that $U_i$ can be obtained by $(r, \beta)$-repair from $\{U_1,
\ldots, U_n\} \setminus \{U_i\}$ for any $U_i \in \{U_1, \ldots, U_n\}$.
\end{defn}
We observe that if $(\Sigma,\A)$ is an exact repair code, then any for any repairable set $R=\{U_1,U_2,\dotsc,U_n\}$, the collection $\A^\prime=\{R\setminus \{U_i\}|1\leq i\leq n\}$ has the property that $(\Sigma,\A^\prime)$ is itself an exact repair code.

\begin{defn}[Strictly functional repair]\label{defn:strictlyfunctional}
Let $(\Sigma={\rm PG}(m-1, q), \A)$ be an $(m; n, k,r,\alpha,\beta)$-functional
repair code.  We say that $(\Sigma, \A)$ is a strictly functional
repair code if there exists a repairable set $\{U_1, \ldots, U_n\}$ for which there is a $U_i \in \{U_1, \ldots, U_n\}$ that cannot be
obtained from $\{U_1, \ldots, U_n\} \setminus \{U_i\}$ by $(r,
\beta)$-repair.
\end{defn}
In other words, $(\Sigma,\A)$ is a strictly functional repair code if there is some subspace in a repairable set such that exact repair from the remaining $n-1$ subspaces of the set is not possible.  For these definitions we are focussing on the subspaces stored by the nodes, rather than explicitly referring to bases for these spaces.  This is due to the fact that the elements stored by a node allow them to recover any desired element in the corresponding space, and this ability does not depend on the choice of basis used to describe the space.   We note that
in \cite{subspacecode}, the term {\em functional repair} is used in a scenario in which the failed node and the repaired node correspond to different bases of the same space.  However, this would satisfy Defintion~\ref{defn:exact} for exact repair, and hence would not represent a strictly functional repair code according to our usage of terminology in this paper.  We will later discuss two examples of codes that do satisfy our stronger definition of strictly functional repair: one
from \cite{HollmannPoh} (Section \ref{sec:strictly}) and a new example that
arises almost immediately from phrasing the definition in terms of
projective geometry (Section \ref{sub:dualarcs}).

\subsection{Geometric interpretation of the cut-set bound}

In \cite{NetworkCoding}, the cut-set bound of network coding is used
to establish an upper bound on the the number of information symbols
$m$ that can be stored in an $(m; n, k,r,\alpha,\beta)$-functional
repair code.  Here we interpret this bound in terms of
finite projective geometry for the case $n=r+1$, $\beta=1$.

\begin{thm} \label{thm:cutset}
Let $(\Sigma={\rm PG}(m-1,q),{\cal A})$ be a $(m;r+1,k,r,\alpha,1)$-functional repair code.  Then 
\begin{align*}
m\leq \sum_{i=1}^k\min (\alpha, (r-k)+i).
\end{align*}
\end{thm}

\begin{proof}
Each node $i$ corresponds to a subspace $U_i$ of $\Sigma$ of dimension
$\alpha-1$, and any $k$ of them span ${\rm PG}(m-1,q)$.  In
particular, the spaces corresponding to the first $k$ nodes span
$\Sigma$, i.e. $\langle U_1,U_2,\dotsc, U_k\rangle=\Sigma$.  This
implies that $m-1$ is at most $k\alpha-1$.

Consider a repair of node 1.  The repair property implies it is
possible to choose one point $P^1_j$ from each node $j$ with $2\leq j
\leq r+1$ such that there is an $(\alpha-1)$-dimensional subspace
$U_1'$ contained in their span with $\{U_1',U_2,\dotsc,U_{r+1}\}$ repairable.  Since we require $\langle U_1^\prime,U_2,\dotsc,U_k\rangle=\Sigma$, it follows that $\langle U_2,U_3,\dotsc,U_k,
P^1_{k+1},P^1_{k+2},\dotsc, P^1_{r+1}\rangle=\Sigma$.  This implies
that $m-1$ is at most $(k-1)\alpha-1+(r+1-k)$.

We now consider a repair of node 2.  There exists a point $P^2_j$ in
each node with $j\neq 2$ (including $P^2_1$ in $U_1'$) such that there
is a $(\alpha-1)$-dimensional subspace $U_2'$ contained in their span with $\{U_1',U_2',\dotsc,U_{r+1}\}$ repairable,
and $\langle U_3,\dotsc,U_k, P^1_{k+1},P^1_{k+2},\dotsc,
P^1_{r+1},P^2_1,P^2_{k+1},P^2_{k+2},\dotsc,P^2_{r+1}\rangle=\Sigma$.
This implies that $m-1$ is at most $(k-2) \alpha-1+(r+1-k)+(r+2-k)$.

We can repeat this process, continuing to replace each $U_i$ in the
set by a collection of repair points whose inclusion ensures that the
replacment $U_i'$ will be contained in the relevant span.  After
repair of node $i$ we have the result that $m-1$ is at most
$(k-i)\alpha-1+\sum_{j=1}^i(r+j-k)$.  The bound on $m-1$ is lowered at
each step until either we reach a point at which the number of
additional points we have to add ($r+i-k$) is greater than $\alpha$,
or we have replaced replaced all of $U_1,\dotsc, U_k$ with the
relevant repair sets of points. At this point we stop, and we have
\begin{align*}
m-1&\leq \left(\sum_{j=1}^k \min(\alpha,r+j-k)\right)-1,\intertext{so}
m&\leq \sum_{j=1}^k \min(\alpha,r+j-k),\\
&=\sum_{i=0}^{k-1}\min(\alpha,r-i).
\end{align*}

\end{proof}

Generalising to $\beta>1$ is entirely straightforward: in each step of
the proof we take $\beta$ points per node rather than $1$ point.  
This approach would also work in the $n>r+1$ case if we make the 
assumption that {\em any} set of $r$ nodes can be used for repair.  This is the assumption made in \cite{NetworkCoding, productmatrix}.

\subsection{Performance measures for FRCs}
In the definitions of Section~\ref{sec:general} there is no stipulation on the size of $\A$, nor
on the number of $(\alpha-1)$-dimensional subspaces in an $(m; n,
k,r,\alpha,\beta)$-functional repair code.  Let $\Num$ be the number
of distinct $(\alpha-1)$-dimensional subspaces used in $\A$.  We will
consider bounds on the value of $\Num$ in Section \ref{sec:digraphs}.

The commonly-studied measures of efficiency of an FRC are the {\em
  storage rate} $R_s = \frac{m}{n \alpha}$ (the number of message
symbols divided by the total number of stored symbols)and the {\em
  repair rate} $R_r=\frac{\alpha}{r \beta}$ (the number of symbols
required for the repaired node divided by the number of symbols
requested in order to facilitate repair).  The value $r \beta$ is
called the {\em repair bandwidth}. Another performance metric is
{\em locality} - the number of nodes to be contacted for repair, given
by $r$.

Other performance metrics that we will not describe formally include
{\em availability}, which is the number of disjoint repair
sets for a node.  Recent interest in
this includes \cite{SahraeiGastpar} where fractional repetition codes
are used to construct codes with high availability and nodes are
partitioned into clusters, each cluster providing a set of helper
nodes to repair a failed node, and \cite{SohnChoiMoon}, where 
codes with different repair bandwidth for repair within clusters and 
across clusters are proposed.

The ability to repair multiple failures is also
obviously of interest, and this may also be studied under different
models, for example, \cite{YeBarg, ZorguiWang} study centralised
repair (where repair is carried out in one location) and cooperative
repair (where failed nodes may communicate) for multiple failures.

Much existing literature seeks to construct codes that optimise one or
more of these measures
(\cite{NetworkCoding,productmatrix,local-avail}).  This is not the
primary motivation of this paper, although we will examine the
trade-offs that arise from the various possible construction choices
we discuss.  We will see that most geometrical constructions seem to
have good repair rates but less than ideal storage rates; some of them
offer a trade-off between repair rate and locality.

\subsection{The product-matrix model}\label{sub:motivation}
The other widely used model of $(m; n, k,r,\alpha,\beta)$-functional repair codes
is the product-matrix model \cite{productmatrix} mentioned in the Introduction.  In this model, the $m$ 
information symbols are formatted into an $r \times \alpha$ message matrix,
and the encoding process involves multiplication by an $n \times r$ encoding
matrix.  The resulting $n \times \alpha$ matrix gives the symbols stored on
each of the $n$ nodes: row $i$ of the matrix denotes the $\alpha$ symbols
stored in node $i$.  This can be viewed as an instantiation of the vector space model of \cite{HollmannPoh}: if the entries in the $i^{\rm th}$ row of the encoding matrix are $E_{i1},E_{i2},\dotsc E_{ir}$, then the $i^{\rm th}$ node corresponds to the subspace spanned by the vectors $\bf{v_0},\bf{v_1},\dotsc, \bf{v_{\alpha-1}}$, where $\mathbf{v_j}$ has the values $E_{i1},E_{i2},\dotsc, E_{ir}$ in positions $j r+1$ through $j r+r$ and $0$ in the remaining positions.  If a length $m$ message is obtained by concatenating the columns of the message matrix, then the resulting symbols stored by each node according to this vector space scheme are precisely those that would be stored using the product-matrix model.

\subsection{Subpacketisation/vectorisation}

We now consider a well-known example of an FRC that can be generated using the the product-matrix model with $\alpha=1$, together with the application of a technique proposed by Shanmugam \textit{et al.} for improving the repair bandwidth \cite{scalar}.  We will see that this example can be described very naturally in the projective geometry setting. 
 
\begin{eg}[Scalar MDS code]
A file $x_0 \ldots
x_{m-1}$ consisting of $m$ symbols belonging to the field $\F_{p^s}$,
$p$ a prime power and $s>1$, is stored across $n$ storage nodes using
an $[n,m]$-MDS code over $\F_{p^s}$. (This is referred to as a {\em scalar
MDS code.})  Each storage node would then store exactly $\alpha=1$
symbol of $\F_{p^s}$.  Now, if one storage node should fail, a repair
would involve contacting $r=m$ nodes, each contributing $\beta=1$
symbol.  Altogether it would take $r \beta = m$ symbols to repair one
symbol.

Following the approach of Definition \ref{defn:functional}, the scalar MDS code construction
translates to a collection of $n$ points $P_0, \ldots, P_{n-1}$ in
$\Sigma={\rm PG}(m-1,p^s)$, every $m$ of which span $\Sigma$;  this is precisely an $n$-arc in ${\rm PG}(m-1, p^s)$. Any failed node can only be obtained by a $(m, 1)$-repair, since any given point of the arc is not contained in the space spanned by $m-1$ further points of the arc.  This is an
$(m; n,m,m,1,1)$-functional repair code with storage rate $R_s = \frac{m}{n}$ and repair rate $R_r=\frac{1}{m}$.

\end{eg}

In \cite{scalar} Shanmugam \textit{et al.} proposed a
``vectorisation'' of MDS codes over fields of prime power in order to
obtain a better repair bandwidth.  ``Vectorisation'' or
``subpacketisation'' involves treating each symbol $x_i\in \F_{p^s}$ as $s$ symbols of
$\F_{p}$.  As a consequence, instead of having to downloading {\em all} the symbols
in each node, one may be able to effect repair by downloading fewer
symbols (from perhaps more nodes), resulting in a reduction of repair
bandwidth.

To explore the vectorisation process more explicitly, let $f(x) = a_0 + a_1x + \cdots + a_{s-1}x^{s-1} +
x^s$ be a primitive polynomial of degree $s$ over $\F_p$
and let $\zeta$ be a root of $f(x)$.  Then every element $b$ of
$\F_{p^s}$ can be written as $b = b_0 + b_1 \zeta + \cdots +
b_{s-1}\zeta^{s-1}$, $b_i \in \F_p$.  Using this correspondence, $b
\in \F_{p^s}$ can be viewed as $(b_0, b_1, \ldots, b_{s-1}) \in
\F_p^s$.  This is the basis of the technique of field reduction used
to construct Desarguesian spreads of ${\rm PG}(sm-1, p)$ from the points of
${\rm PG}(m-1, p^s)$ (\cite[Section 4]{Hirschfeld}).  A point 
$(x_0, x_1, \ldots, x_{m-1})$ in ${\rm PG}(m-1,p^s)$, with 
$x_i \in \F_{p^s}$ viewed as $(x_0^i, x_1^i, \ldots, x_{s-1}^i) \in 
\F_p^s$, can be written as the point 
$(x_0^0, x_1^0, \ldots, x_{s-1}^0, x_0^1, x_1^1, \ldots, x_{s-1}^1,
\ldots,x_0^{m-1}, x_1^{m-1}, \ldots, x_{s-1}^{m-1})$ 
in ${\rm PG}(sm-1, p)$.
Now, take a point $(p_0, p_1, \ldots, p_{m-1}) \in
{\rm PG}(m-1, p^s)$ and all its multiples $\{ (p_0\zeta^i, p_1 \zeta^i,
\ldots, p_{m-1} \zeta^i \; | \; i=0, \ldots, p^s-2 \}$.  Then the
corresponding points of this set in ${\rm PG}(sm-1, p)$ form an
$(s-1)$-dimensional subspace.  The set of all such $(s-1)$-dimensional
subspaces partitions ${\rm PG}(m-1, p^s)$ and is a Desarguesian spread.

(The ``vectorisation'' process in \cite{scalar} uses another map:
each $b \in \F_{p^s}$ can be treated as a linear transformation 
$x \mapsto bx$ in $\F_{p^s}$, so $b$ can be described as an $s \times s$
matrix acting on the basis of $\F_{p^s}$ over $\F_p$.  Each element of 
the MDS code is thus replaced by its corresponding $s \times s$ matrix.
This process is equivalent to the field reduction construction of 
Desarguesian spreads described above.) 

The ``vectorised'' functional repair code is now an
$(sm; n, m, r \le m, s, \beta)$-functional repair code for some $r$ and $\beta$
and storage rate $R_s = \frac{m}{n}$, repair rate $R_r=\frac{s}{r \beta}$.
It corresponds to a set of $n$
$(s-1)$-dimensional subspaces of ${\rm PG}(sm-1, p)$, and we can see that
with more room to manoeuvre we may be able to repair one subspace
without having to use entire subspaces.

We give a small example to illustrate this principle:

\begin{eg} \label{eg:vectorise}
Take $s=3$, $k=3$, $n=5$, we have a 5-arc in 
${\rm PG}(2,8)$ (taking primitive element $\zeta^3 = \zeta + 1$):
\begin{center}
$\left( \begin{array}{ccc}
1 & 0 & 0 \\
0 & 1 & 0 \\
0 & 0 & 1 \\
1 & 1 & 1 \\
1 & \zeta & \zeta^2 \end{array} \right).$
\end{center}

This is an $(m=3; n=5, k=3, r=3, \alpha = 1, \beta = 1)$-functional repair
code with $R_s = \frac{3}{5}$, $R_r= \frac{1}{3}$.  ``Vectorisation'' 
gives 5 planes in ${\rm PG}(8,2)$: each group of three rows are 3 points that span 
a plane.  We call the planes $U_1, \ldots, U_5$.

\begin{center}
$\left( \begin{array}{ccccccccccc}
1 & 0 & 0 && 0 & 0 & 0 && 0 & 0 & 0 \\
0 & 1 & 0 && 0 & 0 & 0 && 0 & 0 & 0 \\
0 & 0 & 1 && 0 & 0 & 0 && 0 & 0 & 0 \\ 
  &   &   &&   &   &   &&   &   &   \\
0 & 0 & 0 && 1 & 0 & 0 && 0 & 0 & 0 \\
0 & 0 & 0 && 0 & 1 & 0 && 0 & 0 & 0 \\
0 & 0 & 0 && 0 & 0 & 1 && 0 & 0 & 0 \\
  &   &   &&   &   &   &&   &   &   \\
0 & 0 & 0 && 0 & 0 & 0 && 1 & 0 & 0 \\
0 & 0 & 0 && 0 & 0 & 0 && 0 & 1 & 0 \\
0 & 0 & 0 && 0 & 0 & 0 && 0 & 0 & 1 \\
  &   &   &&   &   &   &&   &   &   \\
1 & 0 & 0 && 1 & 0 & 0 && 1 & 0 & 0 \\
0 & 1 & 0 && 0 & 1 & 0 && 0 & 1 & 0 \\
0 & 0 & 1 && 0 & 0 & 1 && 0 & 0 & 1 \\
  &   &   &&   &   &   &&   &   &   \\
1 & 0 & 0 && 0 & 0 & 1 && 0 & 1 & 0 \\
0 & 1 & 0 && 1 & 0 & 1 && 0 & 1 & 1 \\
0 & 0 & 1 && 0 & 1 & 0 && 1 & 0 & 1 \\
\end{array} \right).$
\end{center}

This is now an $(m=9; n=5, k=3, r=5, \alpha = 3, \beta = 2)$-functional
repair code.  If $U_1$ fails, one could repair $U_1$ by 
downloading the following points:
\begin{itemize}
\item $R_{21}=(000 \; 110 \; 000)$, $R_{22}=(000 \; 011 \; 000)$ from $U_2$,
\item $R_{31}=(000 \; 000 \; 110)$, $R_{32}=(000 \; 000 \; 011)$ from $U_3$,
\item $R_{41}=(110 \; 110 \; 110)$, $R_{42}=(011 \; 011 \; 011)$ from $U_4$,
\item $R_{51}=(010 \; 101 \; 011)$ from $U_5$ (and another one if we must have
symmetry).
\end{itemize}

Then we can get $(010 \; 000 \; 000) = R_{51}+R_{21}+R_{22}+R_{32}$,
$(110 \; 000 \; 000)= R_{41}+R_{21}+R_{31}$, and
$(011 \; 000 \; 000)= R_{42}+R_{22}+R_{32}$.  This gives us $U_1$.

In the scalar version, to repair one point (9 bits of information) we
need to use three points (27 bits).  The repair rate is therefore $1/3$.  
In the ``vectorised'' version, to repair one subspace (27 bits) we need 
to use 8 points (72 bits).  The repair rate is thus $3/8 > 1/3$. (Or $3/7$
if we don't mind lopsidedness.)  
\end{eg}

The motivation in \cite{scalar} is to obtain a better repair rate, which
the example illustrated.  
In addition, we see that this process has a natural
counterpart in projective geometry that is also intuitive.

Much work has been done further along these lines with some
variations.  For instance, \cite{Babu-Kumar} studies the lower bound
for $\alpha$ (the ``sub-packetisation'') in MSR codes that allow
``repair-by-transfer'', that is, symbols from the remaining
functioning nodes are downloaded directly without computation during
repair, and \cite{Vajha-Babu-Kumar} provides further examples of codes
reaching the lower bound for $\alpha$ for different values of locality
$d$ ($r$ in this paper).  Meanwhile, \cite{epsilonMSR} studies trading
off repair bandwidth for better sub-packetisation, and
\cite{SmallFields} also provides constructions for MSR codes achieving
the lower bound for $\alpha$ for ``repair-by-transfer''.

\section{Projective geometric constructions of functional repair codes}
\label{sec:constructions}

We will examine some existing constructions and also some
constructions that arise naturally from looking at functional repair
codes from a projective geometric point of view. The construction of a vector space/projective geometric functional repair code involves choosing both the dimensions of the spaces corresponding to the nodes, and selecting which subspaces of these dimensions to use.  The properties of the code are determined entirely by the manner in which the various spaces intersect.  The advantage of the geometric perspective is that many classical geometric objects have nice, well-understood properties in terms of how spaces embedded in these objects intersect.  We will see that many existing constructions in the literature can be viewed as arising from classical geometric objects in this way.

Broadly speaking, assigning low-dimension subspaces over a given field to nodes is efficient from a storage perspective, while assigning larger spaces over the same field can allow the repair bandwidth to be reduced.  When spaces of dimension greater than one are used, there is the potential for the spaces assigned to distinct nodes to have a non-trivial intersection.  In what follows we will consider separately constructions with intersecting subspaces and those with non-intersecting subspaces.  Both cases are potentially of interest: non-intersecting spaces are efficient in the sense of avoiding direct redundancy, however there is an upper bound to how large spaces can be without intersecting, and redundancy may be desirable for facilitating recovery and/or repair.

\subsection{Constructions using intersecting subspaces.}
\label{sub:intersect}
We begin by considering the simplest possible case for intersecting subspaces, that of lines in a plane, then use the results obtained to suggest useful constructions in higher dimensions.
\subsubsection{Dual arcs} \label{sub:dualarcs}

A neat construction of an exact repair code can be obtained from three lines in a plane:

\begin{eg}[Three lines in a plane.] \label{eg:lines}
Any three non-concurrent lines in a plane will give an exact repair code: let
$l_1$, $l_2$, $l_3$ be three non-concurrent lines in ${\rm PG}(2,q)$, and
let $\A$ be the collection of the sets of pairs of distinct lines 
$\{l_i, l_j\} \subseteq \{l_1, l_2, l_3\}$. Then $\A$ is an $(m=3; n=3,
k=2,r=2, \alpha=2, \beta = 1)$-functional repair code.

We may coordinatise $l_1$, $l_2$, $l_3$ as:
\begin{eqnarray*}
l_1 & : & \langle (1,0,0), (0,1,0) \rangle, \\
l_2 & : & \langle (0,1,0), (0,0,1) \rangle, \\
l_3 & : & \langle (1,0,0), (0,0,1) \rangle.
\end{eqnarray*}
A $(2,1)$-repair for $l_3$, for example,  is $\langle (1,0,0)\in l_1,
(0,0,1) \in l_2 \rangle$. 

Here the storage rate $R_s = 1/2$ and the repair rate is $R_r=1$. 
\end{eg}

This example tolerates a single node failure.  In order to protect against additional failures we may desire schemes permitting more nodes.  We can achieve this by generalising the idea of Example~\ref{eg:lines} to a larger set of lines: a {\em dual arc}
in a a projective plane of order $q$ is a set of $\Num \le q+1$ lines, 
no three concurrent. 


\begin{thm}\label{thm:lines}
Let $\LL$ be a collection of $\Num \ge 3$ lines of a projective plane
$\Sigma$ such that no three of them are
concurrent.  Let $\A$ be the collection of $(\Num-1)$-tuples of distinct lines of
$\LL$.  Then $(\Sigma, \A)$ is an $(m=3; n=\Num, k=2, r=2, \alpha = 2,
\beta = 1)$-functional repair code that can tolerate up to $\Num-2$ node failures if $\Num > 3$.
\end{thm}

\begin{proof}
Any subset of three nodes in $\LL$ can be considered to be an exact repair code, as seen in Example~\ref{eg:lines}.  Thus, provided two nodes survive, any failed node can be recovered by exact repair.
\end{proof}

\begin{const}[Dual arcs in a plane]\label{cons:dualarcsinplane}
Let $\cal C$ be a nonsingular conic in ${\rm PG}(2,q)$ with $q$ an odd prime power.  Let $\LL$ be a subset of the tangents to $\cal C$ with $|\LL|=n$, $3 \le n \le q+1$.  Then $\LL$ is a dual arc in ${\rm PG}(2,q)$ (\cite{Hirschfeld}).  Let $\A$ be the collection of pairs of distinct lines of
$\LL$.  Then by Theorem \ref{thm:lines}, we have that $(\Sigma, \A)$ is an $(3; n, 2, 2, 2,
1)$-functional repair code that can tolerate up to $n-2$ node failures (if $n > 3$), with storage rate $R_s = 3/2n \le 1/2$ and repair rate $1$.
\end{const}

This construction leads naturally
to a generalisation to higher dimensional spaces:

\begin{eg}[Planes in ${\rm PG}(3,q)$.] \label{eg:planes}
Consider a dual arc in ${\rm PG}(3,q)$: a set of $q+1$ planes,
any 4 meeting trivially.  (So 2 planes meet in a line, 3 planes
meet in a point.)

Take 3 of the planes $\pi_1$, $\pi_2$, $\pi_3$.
If $\pi_3$ fails, repair to $\pi'_3$ using lines $l_i \in \pi_i$,
$i=1,2$. This gives a $(m=4; 3 \le n \le q+1, k=2, r=2, \alpha = 3,
\beta = 2)$-functional repair code.

For example, 
\begin{eqnarray*}
\pi_1 & : & x_0=0, \\
\pi_2 & : & x_1=0, \\
\pi_3 & : & x_2=0. 
\end{eqnarray*}
If $\pi_3$ fails, for example, it can be repaired by $(2,2)$-repair
using lines $l_1 = \{ (0, x_1, 0, x_3)\;|\; x_1, x_3 \in \F_q, 
\mbox{ not both zero.} \} \in \pi_1$ and 
$l_2 =\{(x_0,0,0,x_3)\;|\; x_0, x_3 \in \F_q, 
\mbox{ not both zero.}\} \in \pi_2$.  This gives an 
$(m=4; n=3, k=2, r=2, \alpha = 3, \beta = 2)$-functional repair code,
with  $R_s = 4/3n = 4/9$, $R_r=3/4$.  

On the other hand, we could take 4 planes 
\begin{eqnarray*}
\pi_0 & : & x_0=0, \\
\pi_1 & : & x_1=0, \\
\pi_2 & : & x_2=0, \\
\pi_3 & : & x_3=0.
\end{eqnarray*}
If $\pi_3$ fails, it can be repaired by $(4,1)$-repair, using 
$P_0 = (0,1,0,0) \in \pi_0$, $P_1 = (0,0,1,0) \in \pi_1$, and  
$P_2 = (1,0,0,0) \in \pi_2$.
This gives an 
$(m=4; n=4, k=2, r=3, \alpha = 3, \beta = 1)$-functional repair code, with 
$R_s=4/3n=1/3$ and better repair rate, $R_r=1$.  
\end{eg}

There are two important features in the simple construction of Example
\ref{eg:planes}: the ability
to trade off locality and repair bandwidth without having to make a 
decision during the set up, and the ability to repair multiple failures.
Before we discuss this in more detail, we give the general construction:

\begin{const} \label{cons:dualarc}  
Take a dual arc in ${\rm PG}(m-1, q)$: a set of $q+1$
hyperplanes, any $m$ of which meet trivially.  We may take the set of hyperplanes in a dual normal rational curve $\{ H_t = [1, t, t^2, \ldots, t^{m-1}] \; : \; 
t \in \F_q\} \cup \{H_{\infty} = [0,0, \ldots, 0, 1]\}$, where
$[z_0, z_1, \ldots, z_{m-1}]$ denotes the set of points 
$\{(x_0, x_1, \ldots, x_{m-1}) \; | \; z_0 x_0 + z_1 x_1 + \cdots + z_{m-1}x_{m-1}=0\}$. 

However, to make the description of the trade-off clearer, we will take an 
$m$-subset of these hyperplanes and coordinatise them as follows, writing $e_i$ to denote the point with a 1 in position $i$ and 0
everywhere else:

$$H_i :  x_i=0, \mbox{ that is, } 
H_i=\langle e_j, \;|\; j \in \{0, \ldots, m-1 \}\setminus \{ i \} \rangle. $$

This gives an $(m; n=m, k=2, r=\lceil \frac{m-1}{\beta}
\rceil, \alpha = m-1, \beta)$-functional repair code
with $R_s=\frac{m}{n\alpha}$ and $R_r=\frac{m-1}{m-1+\delta}$, 
where $\delta = 0$ if $\beta|m-1$.  Otherwise $\delta = \beta - \Delta$
where $\Delta = m-1 \bmod{\beta}$.  Here $\beta \ge 1$
and $r \ge 2$.  Indeed, if we choose $m$ odd, and $\beta = (m-1)/2$,
then we achieve both minimum locality and optimum repair bandwidth.

For simplicity we describe what happens if $H_0$ fails.  An $(r,
\beta)$-repair can be performed, with $r = \lceil \frac{m-1}{\beta}
\rceil$, with each of the active $H_i$ contributing $\beta$ points as
follows:

\begin{eqnarray*}
H_1 & \rightarrow & e_2, \ldots, e_{\beta+1}, \\
H_2 & \rightarrow & e_{\beta+2}, \ldots, e_{2\beta+1}, \\
 & \vdots & \\
H_i  & \rightarrow & e_{(i-1)\beta+2}, \ldots, e_{i\beta+1}, \\
 & \vdots & \\
H_{\lceil \frac{m-1}{\beta}\rceil}  & \rightarrow & 
e_{(\lceil \frac{m-1}{\beta}\rceil-1) \beta + 2}, \ldots, e_{m-1}, e_1.  
\end{eqnarray*}

\end{const}

Clearly at any repair one could choose the locality $r$ to suit the
circumstances.  In \cite{subspacecode} a construction was given that
also allows such a trade-off - one can choose between minimum
bandwidth repair or low locality repair, by assigning the subspaces
accordingly, but this assignment has to be determined at set up.  Construction \ref{cons:dualarc} allows the trade-off to be performed at each repair
according to the network conditions.

Construction \ref{cons:dualarc} also tolerates multiple node failures:  we can
choose $ m \le n \le q+1$, and any failure of up to $n-2$ nodes still
allows recovery and repair.  It also gives high availability.  For
example, when we consider the special case of Construction \ref{cons:dualarcsinplane} using dual arcs in planes, we see that any
line can be repaired using any pair of lines,  so that 
many sets of nodes can be used to repair a failed node.  

We note also that if we start with $n < q+1$, additional nodes can be created by
accessing information from some existing nodes using the repair
process.  This may be useful if resilience requirements change during
the lifetime of the storage system.

\subsubsection{Concurrent lines and strictly functional repair}
The use of dual arcs in constructing functional repair codes is appealing due to the high availability that results.  However
Example \ref{eg:lines} also prompts another question:  what happens if we 
allow sets of nodes that correspond to concurrent lines?  In Constructions \ref{cons:dualarcsinplane} and \ref{cons:dualarc}, the spaces assigned to nodes correspond to hyperplanes in the underlying space.  This means their pairwise intersections are well understood: any two hyperplanes of ${\rm PG}(m-1,q)$ intersect in a space of dimension $m-3$.  The use of the dual arcs enables us to control the way in which the spaces corresponding to larger sets of nodes intersect: any $t$ of them intersect in a space of dimension $m-1-t$.  However, we may wish to consider constructions where more general patterns of intersection are allowed (for example, in order to permit more than $q+1$ nodes).

In order to explore what happens when more general patterns of intersection occur, we return to the case of lines in the plane, and consider collections of lines that include sets of three concurrent lines.  The following example shows this takes us into the realm of
strictly functional repair codes:

\begin{eg}[A strictly functional repair code.]\label{eg:concurrent}
Let $l_1, l_2, l_3, l_4$ be four lines of $\Sigma={\rm PG}(2,q)$, $q> 3$,
such that $l_1, l_2, l_3$ are concurrent at a point $P$, and $l_4$
does not pass through $P$. (See Figure
\ref{fig:concurrent}.)  Let $\A$ be the collection of
pairs of lines $\{l_i, l_j\}$, $i, j \in \{1,2,3,4\}$, $i \neq j$.
Then $(\Sigma, \A)$ is am $(m=3; n=3, k=2, r=2, \alpha=2,
\beta=1)$-functional repair code which is strictly functional repair
code.

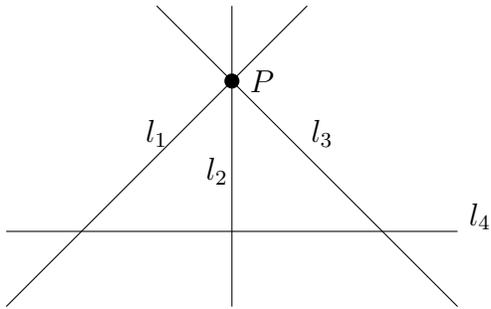
\begin{figure} 
\caption{A strictly functional repair code in ${\rm PG}(2,q)$.}\label{fig:concurrent}
\begin{center}
\begin{tikzpicture}[scale=1]
\draw  (0,0) -- (4,4); 
\draw  (3,0) -- (3,4); 
\draw  (6,0) -- (2,4); 
\draw  (0,1) -- (6,1); 
\fill[black] (3,3) circle (0.1cm); 
\node at (3.4,3) {$P$};
\node at (2,2.3) {$l_1$};
\node at (2.8,1.8) {$l_2$};
\node at (4.2,2.3) {$l_3$};
\node at (6.3,1.2) {$l_4$};
\end{tikzpicture}
\end{center}
\end{figure}

This is because there is a set $\{l_1, l_2, l_3\}$ with $\{l_1, l_2\}$,
$\{l_1, l_3\}$, $\{l_2, l_3\} \in \A$ but $l_3$ cannot be obtained from
$\{l_1, l_2\}$ by $(2,1)$-repair.
\end{eg}

As far as we are aware, this appears to be the only other example of a
strictly functional repair code in the literature, apart from an example
 due to \cite{HollmannPoh} that we will discuss in Section \ref{sub:focal}.
 
\subsubsection{Grassmann varieties}
\label{sub:grassmann}

The constructions we discussed in Section~\ref{sub:dualarcs} all involve subspaces that are hyperplanes of the ambient space.  This represents one extreme point of the possible trade-off between low repair bandwidth and flexibility of repair at the cost of high storage.  Using smaller dimensional spaces both reduces the storage overhead, and allows for greater flexibility in terms of the size of pairwise intersections between the spaces.  In this environment where greater flexibility is possible, this implies that the spaces must be chosen carefully to achieve the desired intersection properties.  Here we consider an example of a construction from \cite{subspacecode}.  It uses subspace codes constructed from Grassmann varieties in vector spaces.  We will describe it from the point
of view of projective geometry, in order to see how known properties of Grassman varieties make it possible to choose collections of subsets with suitable intersections.

Let $b \ge 2$ and $t \le b$ be integers.  Consider $\Pi_t$, a
$t$-dimensional projective subspace of ${\rm PG}(b, q)$.  Let the points
$X_0, \ldots, X_t$ be a basis for $\Pi_t$.  Write 
$X_i = (x_0^i, x_1^i, \ldots, x_b^i)$ and let $M_{\Pi_t}$ be the $(t+1)
\times (b+1)$ matrix

$$ M_{\Pi_t} = 
\left(\begin{array}{c} X_0 \\ X_1 \\ \vdots \\ X_t \end{array} \right)
= \left( \begin{array}{cccc} 
 x_0^0 & x_1^0 & \ldots & x_b^0 \\
x_0^1 & x_1^1 & \ldots & x_b^1 \\
\vdots & \vdots &     & \vdots \\
x_0^t &  x_1^t & \ldots &  x_b^t \end{array}\right).$$

Write $M_{\Pi_t}(i_0, \ldots, i_t)$ to denote the $(t+1) \times (t+1)$
submatrix of $M_{\Pi_t}$ consisting of columns $i_0, \ldots, i_t$.  
Let $\cal{V}$ be the set of ${b+1 \choose t+1}$ subsets $\{i_0,
\ldots, i_t\}$ of $\{0, 1, \ldots, b\}$, ordered in some way.  
Let $\phi(M_{\Pi_t}(i_0, \ldots, i_t))$ be defined as
$\det(M_{\Pi_t}(i_0, \ldots, i_t))$.  Then
$\phi(M_{\Pi_t})$ is defined as a point in ${\rm PG}(B, q)$, where 
$B = {b+1 \choose t+1}-1$, and the $j$th position of
$\phi(M_{\Pi_t})$ is $\phi(M_{\Pi_t}(i_0, \ldots, i_t))$ with 
$\{i_0, \ldots, i_t\}$ in the given order in $\cal{V}$.

For example, take $t=1$, $b=3$. Suppose $\Pi_1$ is a line in ${\rm PG}(3,q)$
with basis points $(x_0, x_1, x_2, x_3)$, $(y_0, y_1, y_2, y_3)$, and
$$ M_{\Pi_1} = \left( \begin{array}{cccc} 
x_0 & x_1 & x_2 &  x_3 \\
y_0 &  y_1 & y_2 & y_3 \\
\end{array} \right).$$
Then $\phi(M_{\Pi_1})$ is a point in ${\rm PG}(5,q)$ given by 
$$( x_0y_1 - x_1y_0, x_0y_2 - x_2 y_0, x_0y_3 - x_3y_0, x_1y_2 - x_2y_1,
x_1y_3 - x_3y_1, x_2y_3 - x_3 y_2).$$

We call these Grassmann coordinates (or Pl\"{u}cker coordinates, when
$t=1$).  The set of points in ${\rm PG}(B,q)$ corresponding to all the
$t$-dimensional subspaces of ${\rm PG}(b, q)$ is called the
Grassmannian, or the Grassmann variety of the $t$-spaces of
${\rm PG}(b,q)$.  We will concentrate on the case $t=1$ here and refer the
reader to \cite[Chapter 24]{HirschfeldThas} for more details and for
the general case.

For $t=1$, the lines of ${\rm PG}(b, q)$ are mapped to points of ${\rm PG}(B,q)$,
$B = {b+1 \choose 2}-1$.  The $q^2+q+1$ lines lying on a plane in
${\rm PG}(b,q)$ are mapped to a plane in ${\rm PG}(B, q)$ - the collection of such
planes in ${\rm PG}(B, q)$ are called the Greek spaces.  The $q^{b-1} +
q^{b-2} + \cdots + b + 1$ lines through a point in ${\rm PG}(b,q)$ are
mapped to a $(b-1)$-dimensional subspace in ${\rm PG}(B,q)$ - the collection
of such subspaces are called the Latin spaces.  Two Latin (Greek)
spaces meet in at most one point, and a Latin and a Greek space meet
in either a line or the empty set.  If there are three distinct Latin
(Greek) spaces $\pi$, $\pi'$, $\pi''$ such that their pairwise
intersections are distinct points, then any other Latin (Greek) space
$\bar{\pi}$ having distinct points in common with $\pi$ and $\pi'$
will also has a point in common with $\pi''$.  These properties allow
the construction of the functional repair codes described in
\cite{subspacecode}.

\begin{const}[Grassman variety construction \cite{subspacecode}]

The storage nodes $V_0, \ldots, V_{n-1}$ are associated with points
$P_0, \ldots, P_{n-1}$ in ${\rm PG}(b, q)$.   Each point $P_i$ can be 
associated with a collection of lines through that point, which, 
in turn, gives a $(b-1)$-dimensional subspace $M_i$ in ${\rm PG}(B, q)$.
The recovery and repair properties then depend on how the points $P_i$
are chosen:  every $b$ of the $M_i$ should span ${\rm PG}(B, q)$, and
if an $M_i$ should fail, one should be able to obtain it by some 
$(r, \beta)$-repair.  In \cite{subspacecode}, it is shown that
this can be a $(b,1)$-repair or a $(c, b)$-repair for any $c | b$.
This gives an $(m=B+1; n, k=b, r=b, \alpha=b, \beta=1)$-functional repair
code (or an $(m=B+1; n, k=b, r=c, \alpha=b, \beta=b)$-functional repair code
for any $c | b$), 
where $B = {b+1 \choose t+1}-1$, $t \le b$. 

Consider again the example with $t=1$, $b=3$.  One could take
 $n\ge 4$ points in ${\rm PG}(3,q)$ such that no 4 points lie in a plane
(an $n$-arc).  The corresponding Grassmannian would then consist
of $n$ planes in ${\rm PG}(5,q)$ with the property that every pair of 
planes meet in a point, and for any plane, the points of intersection 
with the other $n-1$ planes form an $(n-1)$-arc on the plane.  It is then
clear that any three planes would span ${\rm PG}(5,q)$, while any plane
can be obtained by $(3,1)$-repair.  This gives a repair rate of 1, and a storage rate of $\frac{2}{n} \le \frac{1}{2}$.

\end{const}

\subsubsection{Segre varieties}
\label{sub:segre}
Another class of varieties having subspaces with specific intersection properties are the Segre varieties.  These
can also be used to construct functional repair codes with
intersecting subspaces.  It gives storage rate $R_s= \frac{1}{2}$,
and has some restrictive recovery properties, but may still be
of some interest.

A Segre variety $\SV_{s,t}$ in ${\rm PG}((s+1)(t+1)-1, q)$ is defined as 
follows:

Let $S_t$ be a $t$-dimensional projective space  ${\rm PG}(t,q)$
and  $S_s$ be an $s$-dimensional projective space  ${\rm PG}(s,q)$.
Then 
\begin{eqnarray*}
\SV_{s,t} &=& \{(y_0z_0, y_0z_1, \ldots, y_0z_s;
y_1z_0, y_1z_1, \ldots, y_1z_s; \ldots; y_tz_0, y_tz_1, \ldots, y_tz_s) 
\;| \; \\
& &(y_0, y_1, \ldots, y_t) \in S_t, \; (z_0, z_1, \ldots, z_s) \in S_s
 \}.
\end{eqnarray*}

$\SV_{s,t}$ consists of two opposite systems of subspaces $\Sigma_1$,
$\Sigma_2$: $\Sigma_1$ consists of $q^s+q^{s-1}+\cdots+q+1$ mutually
skew $t$-dimensional subspaces, and $\Sigma_2$ consists of
$q^t+q^{t-1}+\cdots+q+1$ mutually skew $s$-dimensional subspaces.
Each subspace in $\Sigma_1$ meets a subspace in $\Sigma_2$ in exactly
one point.

\begin{eg} \label{eg:segre}
Suppose $s=t=1$.  Then $\SV_{1,1}$ is a hyperbolic
quadric in ${\rm PG}(3,q)$ which consists of $(q+1)^2$ points lying on
$2(q+1)$ lines.  These lines form the two opposite systems of
subspaces, each consisting of $q+1$ mutually skew lines.  If we take
two lines from each system, then if one line fails it can always be
repaired by $(2,1)$-repair from the two lines from the opposite
system.  For recovery, however, we must have $k=2$ lines from the same
system.  The collection of 3-subsets of these 4 lines gives an $(m=4;
n=4, k=2, r=2, \alpha = 2, \beta = 1)$- functional repair code (with
the possibility of adding more nodes by the repair process), with $R_s
= \frac{1}{2}$ and $R_r = 1$.

This example illustrates the importance of the assumption of arbitrary
recovery and repair sets in the cut-set bound: Theorem
\ref{thm:cutset} says that $m \le 3$ for $(k, r \alpha, \beta) =
(2,2,2,1)$.  Here we achieve $m=4$, but the pairs of lines
that constitute a recovery set are more restrictive.
\end{eg}

This can be generalised to $\SV_{t,t}$, $t \ge 1$: take $t+1$
$t$-dimensional subspace from $\Sigma_1$, and $t+1$ $t$-dimensional
subspaces from $\Sigma_2$.  Any one subspace may be obtained by $(t+1,
1)$-repair from the $t+1$ subspaces in the opposite system.  For
recovery, as before, we must have $k=t+1$ subspaces from the same
system.  The collection of $(2t+1)$-subsets of these $2t+2$  subspaces
 gives an $(m=(t+1)^2; n=2(t+1), k=t+1, r=t+1, \alpha =
t+1, \beta = 1)$- functional repair code (again, with the possibility
of adding more nodes by the repair process), with $R_s = \frac{1}{2}$
and $R_r = 1$.

\subsection{Constructions using non-intersecting subspaces.}
\label{sub:non-intersect}

\subsubsection{Spreads and partial spreads}
\label{sub:spreads}

Another natural object to look at when one considers projective space
constructions is spreads and partial spreads.

In \cite[Example 2.1]{HollmannPoh} an $(m=4; n=4, k=2, r=3, \alpha=2,
\beta = 1)$-functional repair code is constructed using four mutually
skew lines in ${\rm PG}(3,2)$.  Here we show that the construction works
over $\F_q$ for any $q \ge 2$.  We describe this construction as
elements from a spread in ${\rm PG}(3,q)$, $q \ge 2$.

\begin{thm}
\label{thm:spread}
Let ${\cal{S}}$ be a regular spread in ${\rm PG}(3,q)$.  Let $l_1$, $l_2$, $l_3$ be
three lines of ${\cal{S}}$ and let ${\cal{R}}$ be the (unique) regulus 
containing them.  Let $l_4 \in {\cal{S}} \setminus {\cal{R}}$.  Then $l_{i_4}$
can be obtained from $l_{i_1}$, $l_{i_2}$, $l_{i_3}$ by $(3,1)$-repair, 
$\{i_1, i_2, i_3, i_4\} = \{1,2,3,4\}$.
\end{thm}

\begin{proof}
It is clear that any three of $l_1, \ldots, l_4$ are contained in a regulus that
does not contained the fourth line, so without loss of generality it suffices to prove that $l_4$
can be obtained from $l_1$, $l_2$, $l_3$ by $(3,1)$-repair.  

Let $Q_1$ be any point on $l_4$.  Let $l_5$ be the transversal through
$Q_1$ to $l_2$, $l_3$ - this line exists and is unique.  Let $P_2 = l_5 \cap l_2$
and $P_3=l_5 \cap l_3$.

Now consider $\{l_1, l_3, l_4\}$.  There is a unique regulus containing them
but not $l_2$.  Let $l_6$ be the transversal to them through $P_3$.  Let
$P_1 = l_6 \cap l_1$ and $Q_2 = l_6 \cap l_4$.  (We know that $Q_1 \neq Q_2$
since otherwise $l_5=l_6$ and $l_6$ meets all four lines, which means all four lines
are in a regulus.)

Now consider the space spanned by $P_1$, $P_2$, $Q_1$, $Q_2$,  $\pi = 
\langle P_1, P_2, Q_1, Q_2 \rangle$.  Since $P_2 Q_1 \cap P_1 Q_2 = P_3$,
$\pi$ is a plane.  So $P_1P_2$ and $l_4$ are both lines in $\pi$ and therefore
$P_1P_2$ meets $l_4$ in a point $Q_3$.  Hence $l_4 \subseteq \langle
P_1 \in l_1, P_2 \in l_2, P_3 \in l_3 \rangle$ and thus is obtained from 
$l_1$, $l_2$, $l_3$ by $(3,1)$-repair.
\end{proof}

\begin{const}{\cite[Example 2.1]{HollmannPoh}}
The collection of pairs of distinct lines from $\{l_1, l_2, l_3, l_4\}$
forms an $(m=4; n=4, k=2, r=3, \alpha=2,\beta = 1)$-functional repair code
which has $R_s = \frac{1}{2}$ and $R_r=\frac{2}{3}$.

For example, we may choose $l_1$, $l_2$, $l_3$ to be 

\begin{eqnarray*}
l_1 & = & \langle (1,0,0,0), (0, 0, 1,1) \rangle, \\
l_2 & = & \langle (0,1,0,0), (1, 0, 0,1)\rangle, \\
l_3 & = & \langle (0,0,1,0), (1, 1, 0,0)\rangle.
\end{eqnarray*}

These are lines on the quadric/regulus 
$$x_0x_2 - x_0x_3 - x_1x_2 - x_2x_3 + x_3^2 = 0.$$

(The other lines of the regulus are $\langle (1,0,y, 1), (1, y, 0, 0) \rangle$.)

We can take $l_4$ to be $\langle (0,0,0,1), (0,1,1,0) \rangle$, which
does not belong to this regulus.
\end{const}

A natural generalisation of such a construction would be to take
planes in spreads in ${\rm PG}(5,q)$.  Indeed, in Section
\ref{sub:motivation} a construction is given using elements of an
$(s-1)$-spread in ${\rm PG}(sm-1,q)$.  In \cite{HouLi, HouLiShum, NamSong,
  OggierDatta}, regular $t$-spreads in ${\rm PG}(m-1,q)$ are used to give
$(m; k \le n \le \frac{2^m-1}{2^{t+1}-1}, k = \frac{m}{\alpha}, r=2,
\alpha = t+1, \beta = \alpha)$-functional repair codes.  These
functional repair codes have the additional property of allowing
repairs of multiple node failures simultaneously.  For example, in
\cite{OggierDatta}, up to $\frac{n-1}{2}$ failed nodes can be repaired
simultaneously.  This follows from the property of regular spreads,
where one can always choose two spread elements that span a subspace that
contains a third given element.

These elements are subsets of a system of subspaces in a Segre
varieties.  Hence it is also natural to consider the generalisation to
subspaces on a Segre varieties.  In contrast to the constructions in
Section \ref{sub:segre} where elements are taken from both systems of
subspaces of a Segre variety, here we only take subspaces from one
system of subspaces, and these are mutually skew. 
Consider again an $\SV_{s,t}$ as described in Section \ref{sub:segre}.
For every point in $S_t$, there is a corresponding $s$-dimensional
subspace belonging to $\Sigma_2$ in $\SV_{s,t}$.  Take a
$t'$-dimensional subspace $V'$ of ${\rm PG}(t,q)$, $t'\le t$, and consider
$\Sigma'$, the $s$-dimensional subspaces contained in $\SV_{s,t}$
corresponding to the points of $V'$. Then, any subspace $W$ in
$\Sigma'$ can be obtained by $(2, s+1)$-repair from two other
subspaces in $\Sigma'$: suppose $W$ corresponds to the point $P \in
V'$, pick a point $P' \in V'$ and another point $P'' \in V'$ collinear
with $P$ and $P'$.  Then the subspaces in $\Sigma'$ corresponding to
$P'$ and $P''$ will span a subspace containing $W$. Let $n =
\frac{q^{t'+1}-1}{q-1}$.  The collection of $(n-1)$-subsets of
$s$-dimensional subspaces from $\Sigma'$ gives an $(m=(s+1)(t+1); n,
k=t+1, r=2, \alpha = s+1, \beta = \alpha)$-functional repair code.

\subsubsection{Focal spreads}
\label{sub:focal}

Let $\Sigma_{2t-1} = {\rm PG}(2t-1,q)$, $t>1$, and let $S_t$ be a
$(t-1)$-spread in $\Sigma_{2t-1}$.  Let $L$ be an element of $S_t$.
Let $\Sigma_{t+d-1}$, $t>d$, be a $(t+d-1)$-dimensional subspace of
$\Sigma_{2t-1}$ that contains $L$.  Then $\{L\} \cup \{ M' = M \cap
\Sigma_{t+d-1} \; | \; M \in S_t \setminus \{L\}\}$ is a focal spread
consisting of the focus $L$, and the $(d-1)$-dimensional subspaces
$M'$ partitioning the points of $\Sigma_{t+d-1}$ not in $L$.  Focal
spreads are described in greater details in \cite{JhaJohnson}.

In \cite{HollmannPoh} an $(m=5; n=4, k=3,
r=3,\alpha=2,\beta=1)$-functional repair code was constructed using
focal spreads with $t=3$, $d = 2$: a 2-spread in ${\rm PG}(5,2)$,
intersected by a 4-space, the focus being a plane, and there are 8
lines partitioning the points not in the plane.  The storage code
consists the collection of 3-subsets of these 8 lines.

This can clearly be generalised.  For example, using $t=4$, $d=2$, we
have the storage code being 16 lines partitioning the set of points of
a 5-dimensional space that are not contained in the focus, which is a
3-dimensional space.  A computer search shows that a line cannot be
obtained by $(3,1)$-repair but can be obtained by $(4,1)$-repair,
making this an $(m=6; n=16; k=3, r=4, \alpha = 2, \beta =
1)$-functional repair code.

However, the example in \cite{HollmannPoh} turns out to be strictly 
functional, while our generalisation allows both functional and exact repair.
Indeed, this appears to be the only strictly functional repair code that is
known (apart from Example \ref{eg:concurrent}).  In the next section we prove this property and examine the
structure further.

\section{Anatomy of a  strictly functional repair code}
\label{sec:strictly}

In \cite[Example 2.2 and Section VI]{HollmannPoh}, an $(m=5; n=4, k=3,
r=3,\alpha=2,\beta=1)$-functional repair code was given which turns out to be
a strictly functional repair code.  This is constructed using focal
spreads and is described in Section \ref{sub:focal}.  Here we prove
that it is strictly functional, and consider whether it can be
generalised.

Firstly we write the $(m=5; n=4, k=3,
r=3,\alpha=2,\beta=1)$-functional repair code according to Definition
\ref{defn:functional}:

\begin{defn} \label{defn:small}
Let $\Sigma = {\rm PG}(4,q)$ and let $\A$ be a set of 3-tuples $\U$ of 
lines such that
\begin{enumerate}[label=(\alph*)]
\item (Recovery) \label{small:rec} For every $\U \in \A$, the 3 lines
  in $\U$ span ${\rm PG}(4,q)$.
\item (Repair) \label{small:rep} For each $\U=\{U_1, U_2, U_3\}$ there
  is a point $P_i$ on $U_i$, $i = 1,2,3$, such that there is another
  line $U_4 \subseteq \langle P_1, P_2, P_3 \rangle$, and $\U_i' = \U
  \cup \{U_4\} \setminus \{U_i\}$, $i=1,2,3$, again belongs to 
$\A$. 
\end{enumerate}
\end{defn}

We will give a brief description of this construction in terms of
projective spaces.  We will describe the lines using the
correspondence between ${\rm PG}(1, 2^3)$ and the spread in ${\rm PG}(5,2)$ in the
manner described in Section \ref{sub:motivation}.

Write $\F_8$ as $\{ 0, \zeta^i : i=0, \ldots, 6, \zeta^3 =
\zeta+1\}$.  If $a = a_0 + a_1 \zeta + a_2 \zeta^2$ and $b = b_0 +
b_1 \zeta + b_2 \zeta^2$ then $(a, b) \in {\rm PG}(1, 2^3)$ can be thought
of as a point $(a_0, a_1, a_2, b_0, b_1, b_2)$ in ${\rm PG}(5, 2)$.  The
point $(a,b) \in {\rm PG}(1, 2^3)$ thus gives a plane $\Pi_{(a,b)}$ in
${\rm PG}(5,2)$ consisting of the points $\{(ax, bx) : x \in \F_8 \}$.
So the point $(1,0) \in {\rm PG}(1, 2^3)$ corresponds to the plane
$$\Pi_{(1,0)} = \langle (1,0,0,0,0,0), (0,1,0,0,0,0), (0,0,1,0,0,0)\rangle.$$
The point $(a,1)$, $a \in \F_8$, corresponds to the plane 
$$\Pi_{(a,1)} = 
\langle (a_0, a_1, a_2, 1, 0, 0), (a_2, a_0+a_2, a_1, 0, 1, 0), 
(a_1, a_1+a_2, a_0+a_2, 0, 0, 1) \rangle.$$

We can take the plane in the focal spread as the plane $\Pi_{(1,0)}$,
and the lines $l_a$ as the intersection of the hyperplane $x_5=0$ with
the planes $\Pi_{(a, 1)}$, $a \in \F_8$.  Treating the hyperplane $x_5=0$ as
${\rm PG}(4,q)$, we may write
$$ l_a = \{(a_0, a_1, a_2, 1, 0), (a_2, a_0+a_2, a_1, 0, 1), 
(a_0+a_2, a_0+a_1+a_2, a_1+a_2, 1,1)\}.$$

Let $\LL = \{l_a : a \in \F_8\}$.  The functional repair code consists of 
the collection of all 3-subsets of $\LL$.
It is not hard to show that any set of 3 lines
$l_a$, $l_b$, $l_c$ from $\LL$ will allow exactly \emph{one} line $l_d
\in \LL$ by $(3,1)$-repair,  and this line satisfies $d^2 = ab + ac +
bc$.  It is also not hard to see that the following two conditions ({\cite[Example 2.2]{HollmannPoh}}) are 
satisfied by the lines of $\LL$:
\begin{enumerate}
\item[(L1)]  Any 3 lines span ${\rm PG}(4,q)$.
\item[(L2)]  Any pair of lines are skew.
\end{enumerate}

This construction works for $q>2$, in the sense that such a
construction for focal spread works over $q>2$, and also a line can be
obtained by $(3,1)$-repair from any three lines (Theorem
\ref{thm:strictlyfunctional}).  However, it is not clear that there is a nice
relationship between $a$, $b$, $c$ and $d$, as in the case for $q=2$.
For example, for the case $q=3$:

Take $x^3-x+1=0$ over $GF(3)$ to get $GF(3^3) = \{0, \alpha^i \;|\;
\alpha^3=\alpha-1\}$.  The point $(a,1)$ on ${\rm PG}(1,3^3)$ with $a = a_0
+ a_1 \alpha + a_2 \alpha^2$ gives the plane 

$$ \langle (a_0, a_1, a_2, 1, 0,0), (-a_2, a_0+a_2, a_1, 0,1,0),
(-a_1, a_1-a_2, a_0+a_2, 0,0,1) \rangle$$

in ${\rm PG}(5,3)$.  Intersecting with $x_5=0$ gives lines $l_a = 
\langle  (a_0, a_1, a_2, 1, 0), (-a_2, a_0+a_2, a_1, 0,1) \rangle$.

We can construct $l_{\alpha^{12}}$ by $(3,1)$-repair from $l_0$, $l_1$
and $l_{\alpha}$, but it is not clear what the relationship between
$a$, $b$, $c$, $d$ is.

\subsection{The focal spread construction is strictly functional}
\label{sub:strictly}
The repair process described above corresponds to functional repair.  In this section we show that this is necessary: this FRC does not admit exact repair.  We begin with a geometric lemma that we will use in the proof of this fact.
\begin{lemma}\label{lem:uniquetransversal}
Let $\{\ell_1,\ell_2, \ell_3\}$ be lines in ${\rm PG}(4,q)$ that satisfy (L1) and (L2).  Then there is a unique line $m$ with $m\cap \ell_i\neq \emptyset$ for $i=1,2,3$.
\end{lemma}
\begin{proof}
By (L2) we know that $\ell_1$ and $\ell_2$ span a hyperplane $\Pi\subset{\rm PG}(4,q)$.  By (L1) we know that $\ell_3$ intersects $\Pi$ in a unique point $P_3$.  Consider the plane $\sigma=\langle P_3, \ell_2 \rangle$.  Since $\ell_1$ and $\ell_2$ span $\Pi$, it follows that $\sigma$ intersects $\ell_1$ in a unique point $P_1$.  The line $m=\langle P_1,P_3\rangle \neq \ell_2$ lies in $\sigma$, as does $\ell_2$, and hence these two lines intersect in a unique point $P_2$.  Thus the line $m$ intersects each of the lines $\ell_1$, $\ell_2$ and $\ell_3$, and it is unique by construction.
\end{proof}

\begin{thm} \label{thm:strictlyfunctional}
Let $\{\ell_1,\ell_2, \ell_3, \ell_4\}$ be lines in ${\rm PG}(4,q)$ that satisfy (L1) and (L2).  Then at most one of the lines can be obtained by exact $(3,1)$-repair from the remaining three lines.
\end{thm}
\begin{proof}
Suppose (without loss of generality) that $\ell_4$ can be obtained by $(3,1)$-repair from $\{\ell_1,\ell_2,\ell_3\}$.  Then there exist points $P_1\in \ell_1$, $P_2\in \ell_2$ and $P_3\in \ell_3$ such that $\ell_4\subseteq \langle P_1,P_2,P_3\rangle$.  We note that it is not the case that $\ell_4= \langle P_1,P_2,P_3\rangle$, for this would imply that $\ell_4=\langle P_1,P_2\rangle$, in which case $\ell_4$ would be contained in $\langle \ell_1,\ell_2\rangle$, in violation of (L1).  Hence $\ell_4\subset \langle P_1,P_2,P_3\rangle$.  The line $\langle P_1, P_2 \rangle $ therefore intersects $\ell_4$ in a unique point, and hence by Lemma~\ref{lem:uniquetransversal} is the unique line $m_{124}$ meeting $\ell_1$, $\ell_2$ and $\ell_4$.  Similarly, $\langle P_1, P_3\rangle $ is the unique line $m_{134}$ meeting $\ell_1$, $\ell_3$ and $\ell_4$.

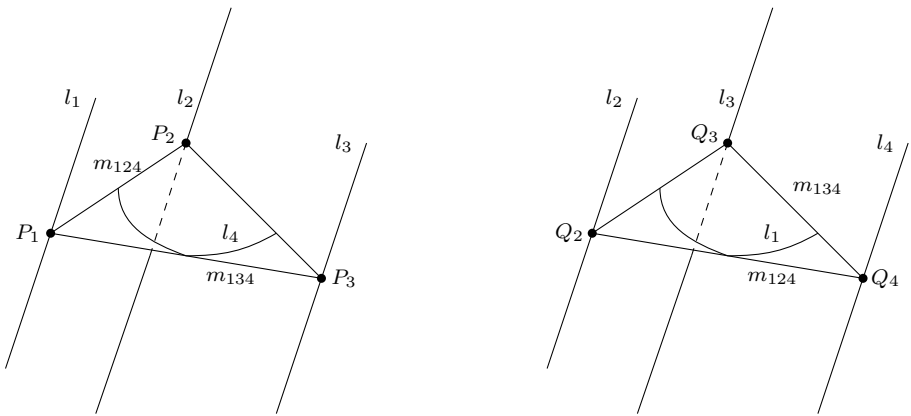
\begin{figure} 
\caption{Repair of  $l_4$ and $l_1$.}\label{fig:l4l1}
\begin{center}
\begin{tikzpicture}[scale=0.6]
\draw  (1,5) -- (4,7) -- (7,4)--(1,5); 
\draw (2.5,6)to [out=270,in=160] (4,4.5) to [out=0,in=210] (6,5); 
\node at (5,5) {{\scriptsize $l_4$}};
\draw (0,2) -- (2,8); \node at (1.5,8) {{\scriptsize $l_1$}};
\draw (2,1) -- (3.2,4.5);  \node at (4,8) {{\scriptsize $l_2$}};
\draw[dashed] (3.2,4.5) -- (4,7);
\draw (4,7) -- (5,10);
\draw (6,1) -- (8,7);  \node at (7.5,7) {{\scriptsize $l_3$}};
\fill[black] (1,5) circle (0.1cm); \node at (0.5,5) {{\scriptsize $P_1$}};
\fill[black] (4,7) circle (0.1cm); \node at (3.5,7.2) {{\scriptsize $P_2$}};
\fill[black] (7,4) circle (0.1cm); \node at (7.5,4) {{\scriptsize $P_3$}};
\node at (2.5,6.5) {{\scriptsize $m_{124}$}};
\node at (5,4) {{\scriptsize $m_{134}$}};

\draw  (13,5) -- (16,7) -- (19,4)--(13,5); 
\draw (14.5,6)to [out=270,in=160] (16,4.5) to [out=0,in=210] (18,5); 
\node at (17,5) {{\scriptsize $l_1$}};
\draw (12,2) -- (14,8); \node at (13.5,8) {{\scriptsize $l_2$}};
\draw (14,1) -- (15.2,4.5);  \node at (16,8) {{\scriptsize $l_3$}};
\draw[dashed] (15.2,4.5) -- (16,7);
\draw (16,7) -- (17,10);
\draw (18,1) -- (20,7);  \node at (19.5,7) {{\scriptsize $l_4$}};
\fill[black] (13,5) circle (0.1cm); \node at (12.5,5) {{\scriptsize $Q_2$}};
\fill[black] (16,7) circle (0.1cm); \node at (15.5,7.2) {{\scriptsize $Q_3$}};
\fill[black] (19,4) circle (0.1cm); \node at (19.5,4) {{\scriptsize $Q_4$}};
\node at (17,4) {{\scriptsize $m_{124}$}};
\node at (18,6) {{\scriptsize $m_{134}$}};
\end{tikzpicture}
\end{center}
\end{figure}

Suppose now that some other line (say, $\ell_1$) can be obtained by $(3,1)$-repair from the remaining lines (i.e. $\{\ell_2,\ell_3,\ell_4\}$).  See Figure~\ref{fig:l4l1}.   Repeating the above argument we observe that there are points $Q_2\in \ell_2$, $Q_3\in \ell_3$ and $Q_4\in \ell_4$ such that $\ell_1\subset \langle Q_2,Q_3,Q_4\rangle$.  However, in this case the line $\langle Q_2,Q_4\rangle$ meets $\ell_1$ in a point, which implies $\langle Q_2,Q_4\rangle=m_{124}$ (by Lemma~\ref{lem:uniquetransversal}), and so $Q_2=P_2$.  Similarly, $\langle Q_3,Q_4\rangle$ meets $\ell_1$ in a point, so $\langle Q_3,Q_4\rangle=m_{134}$, and so $Q_3=P_3$.  But now we have that $Q_2,Q_3,Q_4\in \langle P_1,P_2,P_3\rangle$, and hence $\ell_1\subset \langle P_1,P_2,P_3\rangle$.  This contradicts the fact that $\ell_1$ and $\ell_4$ are not coplanar, by (L2).
\end{proof}

This shows that this focal spread construction is strictly functional:
one can always construct a fourth line $l_4 = m$ from any three lines
$l_1, l_2, l_3$, and if one of $l_1$, $l_2$ or $l_3$ fails, it cannot
be repaired exactly from the three remaining lines.

\section{A simpler description}
\label{sec:simpler}

In our examples and constructions, we could enumerate a set of
subspaces, and simply state that a collection of subsets of these
subspaces constitute a functional repair code, bypassing the recursive
nature of the definition (Definition \ref{defn:functional}).

However, such a description is not always useful, or easy to arrive
at.  Firstly, we would in general like to find small codes.  As an
example, Theorem \ref{thm:lines} allows $\LL$ to be the set of all
lines in a projective plane, but we see in Example \ref{eg:lines} that
3 lines suffices.  Hollmann and Poh \cite[Theorem 5.1]{HollmannPoh}
give a method of starting with a possible set of subspaces $\U =
\{U_1, \ldots, U_{n-1}\}$ and another subspace $U_n$ constructed by
$(r, \beta)$-repair from $\U$, and obtaining a functional repair code
from it using the image under a group action.  In Section
\ref{sec:digraphs} we model this process of building a functional
repair code using digraphs.

Secondly, this kind of description does not always convey the
complications of the repair process.  We illustrate with an example.
The focal spread construction of Section \ref{sec:strictly} admits a  straigtforward description
similar to that of Theorem \ref{thm:lines}: 

\begin{quote}
Let $\LL$ be a set of lines in $\Sigma={\rm PG}(4,q)$ satisfying conditions
(L1), (L2):
\begin{enumerate}
\item[(L1)]  Any 3 lines span ${\rm PG}(4,q)$.
\item[(L2)]  Any pair of lines are skew.
\end{enumerate}
Let $\A$ be a collection
of 3-subsets of $\LL$.  Then $(\Sigma,\A)$ is a functional repair code.
\end{quote}

If we were to want to construction a set of such lines, how would we
start?  Because $\LL$ is a strictly functional repair code (Theorem
\ref{thm:strictlyfunctional}), given a 3-subset $\{l_1, l_2, l_3\}$ in $\A$, we
obtain an $l_4$ by $(3,1)$-repair, but the 3-subset containing $l_4$,
say, $\{l_2,l_3,l_4\}$ will give an $l_5 \neq l_1$ by $(3,1)$-repair.
This motivates the following steps in the construction:

Let $\LL$ be a set of three lines satisfying (L1), (L2) to start with.
\begin{enumerate}
\item Take any 3 lines of $\LL$.  Use $(3,1)$-repair to get a fourth line.
\item Add this fourth line to $\LL$ if it is not already in it.  
\item Repeat until no new lines are constructed.
\end{enumerate}

Take $\A$ to be the 3-subsets of $\LL$.
Then $\A$ is a functional repair code \'{a} la Definition 
\ref{defn:small}.

This motivates a clearer modelling of the repair properties.  We examine
this in the next section.

\section{The repair condition as digraphs}
\label{sec:digraphs}

We write this with $m=5$, $n=4$, $k=3$,
$r=3$, $\alpha=2$ $\beta=1$, for simplicity, 
but it can easily be written more generally.

We can think of the repair condition (Definition
\ref{defn:functional}\ref{pjfn:repair}) of an $(m; n,
k,r,\alpha,\beta)$-functional repair code $(\Sigma,\A)$ as a bipartite
digraph $\graph(\A) = (\vertex(\A) \cup \vertex'(\A), \edge
\cup \edge')$ as follows:

Let $\vertex(\A)$ be a set of vertices corresponding to the sets $\U$
of 3 lines in $\A$ - each set $\U \in \A$ is a vertex in
$\vertex(\A)$.  By the repair condition, one could obtain a fourth
line $U'$ by $(r, \beta)$-repair from any set $\U$ of 3 lines.  Let
$\vertex'(\A)$ be another set of vertices corresponding to these sets
$\U \cup \{U'\}$, $\U \in \A$, of four lines.  The set of vertices of
$\graph(\A)$ will be the (disjoint) union of these two sets
of vertices.

The (directed) edges of $\graph( \Sigma, \A)$ are defined as follows:
There is an edge from $V = \{U_1, U_2, U_3\} \in \vertex(\A)$ to 
$V' = \{U_1, U_2, U_3, U_4 \} \in \vertex'(\A)$ if and
only if $U_4$ is obtain by $(r, \beta)$-repair from $\{U_1, U_2, U_3\}$.
We denote this set of edges by $\edge$.  In addition, there is an edge
from $V' = \{U_1, U_2, U_3, U_4 \} \in \vertex'(\A)$ to $V \in \vertex(\A)$
if and only if $V = V' \setminus \{U_i\}$, $i \in \{1,2,3,4\}$.
We denote this set of edges by $\edge'$.  The set of edges of 
$\graph(\A)$ will be the (disjoint) union of these two sets
of edges.

Clearly there are edges only between $\vertex(\A)$ and $\vertex'(\A)$
and $\graph(\A)$ is a bipartite digraph.  An edge from $\vertex(\A)$
to $\vertex'(\A)$ signifies a repair while an edge from $\vertex'(\A)$
to $\vertex(\A)$ signifies a node failure.  Figure \ref{pic:bipartite}
gives a small example of what the node failures and repairs might look like.

\begin{figure} \caption{$\graph(\A)$ with $n=4$, $k=3$,
$r=3$.}
\label{pic:bipartite}
\begin{center}
\begin{tikzpicture}[scale=1]

\draw (0,6) ellipse (1.5cm and 0.5cm);
\draw (5,6) ellipse (1.5cm and 0.5cm);
\draw (10,6) ellipse (1.5cm and 0.5cm);
\node at (0,6) {$U_1\;U_2\;U_3\;U_4$};
\node at (5,6) {$U_1\;U_2\;U_3\;U_4'$};
\node at (10,6) {$U_1'\;U_2\;U_3\;U_4'$};

\node at (-2,4.5) {$\vertex'(\A)$};
\draw[dashed] (-2,4) --  (13,4);
\node at (-2,3.5) {$\vertex(\A)$};

\draw (2.5,2) ellipse (1cm and 0.5cm);
\draw (7.5,2) ellipse (1cm and 0.5cm);
\draw (12.5,2) ellipse (1cm and 0.5cm);
\node at (2.5,2) {$U_1\;U_2\;U_3$};
\node at (7.5,2) {$U_2\;U_3\;U_4'$};
\node at (12.5,2) {$U_1'\;U_3\;U_4'$};

\draw (1,5.5) edge[out=270,in=120,->] (2,2.5);
\node at (0.5,4.5) {\footnotesize $U_4$ fail};
\draw (6,5.5) edge[out=270,in=120,->] (7,2.5);
\node at (5.5,4.5) {\footnotesize $U_1$ fail};
\draw (11,5.5) edge[out=270,in=120,->] (12,2.5);
\node at (10.5,4.5) {\footnotesize $U_2$ fail};

\draw (3,2.5) edge[out=30,in=270,->] (4.5,5.5);
\node[text width = 1.5cm] at (4,3.5) {\footnotesize{Repair to $U_4'$}};
\draw (8,2.5) edge[out=30,in=270,->] (9.5,5.5);
\node[text width = 1.5cm] at (9,3.5) {\footnotesize{Repair to $U_1'$}};

\end{tikzpicture}
\end{center}
\end{figure}
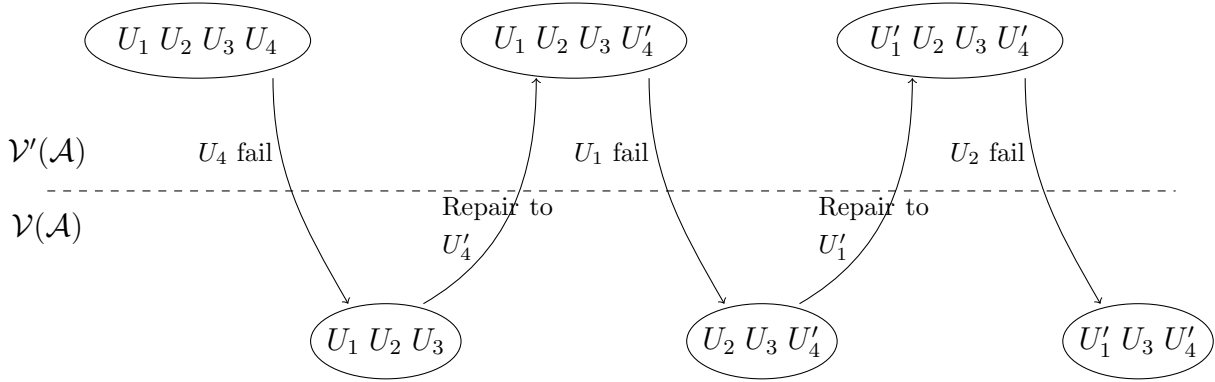

Since each node may fail, there must be four out-edges from each vertex
in $\vertex'(\A)$, and since every three nodes must be able to repair
a fourth node, there must be at least one out-edge from each vertex
in $\vertex(\A)$.

\begin{defn}\label{defn:graph}
Let $\graph = (V_1 \cup V_2, E)$ be a bipartite digraph with parts
$V_1$, $V_2$.  We say that $\graph$ satisfies the repair condition if
all vertices in $V_1$ has outdegree at least 1 and all vertices in
$V_2$ has outdegree $n$.
\end{defn}

This view of a functional repair code immediately gives us some idea on
the number of subspaces we need and the size of $\A$, as well as the
characterisation of exact repair. 

\begin{lemma} \label{lemma:num}
$$|\vertex(\A)| \le {\Num \choose n-1}, \;\;   
|\vertex'(\A)| \le {\Num \choose n}.$$
As a consequence, $\Num \ge n$. 
\end{lemma}

\begin{lemma} \label{lemma:edge}
$$ |\edge(\A)| \ge |\vertex(\A)|, \;\; |\edge'(\A)| = n|\vertex'(\A)|.$$
\end{lemma}

This leads to the characterisation:

\begin{lemma}\label{lemma:exact}
A functional repair code $(\Sigma, \A)$ is an exact repair code 
if and only if $\graph(\A)$ is a complete bipartite digraph (with an
in-edge and and out-edge between each pair of vertices from different parts) 
with $|\vertex(\A)|=n$, $|\vertex'(\A)|=1$.
\end{lemma}

A functional repair code admits exact repair if it has a subgraph
that satisfies the condition in Lemma \ref{lemma:exact}, while
a strictly functional repair code 
would satisfy the condition that there exists $V' \in \vertex'{\A}$,
$V \in \vertex(\A)$, such that $(V', V) \in \edge'(\A)$ but $(V, V')
\not\in \edge(\A)$.

We illustrate this with the strictly functional repair code of Example 
\ref{eg:concurrent}.  Figure \ref{fig:strictly} is the digraph corresponding
to the example.   The dotted lines represent repairs.
The node $\{l_1,l_2,l_3\}$ and the dashed lines show that if any of 
$l_1$, $l_2$ or $l_3$ failed, they cannot be repaired from the remaining
lines.  And if all nodes containing $l_1$ are removed, we have an exact repair code consisting of three non-concurrent lines.

\begin{figure} \caption{$\graph(\A)$} \label{fig:strictly}
\begin{center}
\begin{tikzpicture}[scale=1]
\draw (1,1) ellipse (0.8cm and 0.5cm);
\node at (1,1) {$l_1\;l_2$};
\draw (4,1) ellipse (0.8cm and 0.5cm);
\node at (4,1) {$l_1\;l_3$};
\draw (7,1) ellipse (0.8cm and 0.5cm);
\node at (7,1) {$l_2\;l_3$};
\draw (10,1) ellipse (0.8cm and 0.5cm);
\node at (10,1) {$l_1\;l_4$};
\draw (13,1) ellipse (0.8cm and 0.5cm);
\node at (13,1) {$l_2\;l_4$};
\draw (16,1) ellipse (0.8cm and 0.5cm);
\node at (16,1) {$l_3\;l_4$};
\draw (2.5,4) ellipse (1cm and 0.5cm);
\node at (2.5,4) {$l_1\;l_2\;l_3$};
\draw (6.5,4) ellipse (1cm and 0.5cm);
\node at (6.5,4) {$l_1\;l_2\;l_4$};
\draw (10.5,4) ellipse (1cm and 0.5cm);
\node at (10.5,4) {$l_1\;l_3\;l_4$};
\draw (14.5,4) ellipse (1cm and 0.5cm);
\node at (14.5,4) {$l_2\;l_3\;l_4$};
\draw[dashed] (1.6,3.7) edge[out=200,in=100,->] (1,1.5);
\draw[dashed] (2.5,3.5) edge[out=315,in=100,->] (3.5,1.5);
\draw[dashed] (3.3,3.7) edge[out=315,in=100,->] (6.5,1.5);

\draw[red] (5.6,3.7) edge[out=200,in=100,->] (1.5,1.4);
\draw[red] (6.5,3.5) edge[out=270,in=100,->] (9.5,1.5);
\draw[red] (7.3,3.7) edge[out=315,in=120,->] (12.5,1.5);

\draw[blue] (9.6,3.7) edge[out=200,in=80,->] (4.5,1.4);
\draw[blue] (10.5,3.5) edge[out=315,in=100,->] (10,1.5);
\draw[blue] (11.3,3.7) edge[out=315,in=100,->] (15.5,1.5);

\draw[green] (13.6,3.7) edge[out=200,in=80,->] (7.5,1.4);
\draw[green] (14.5,3.5) edge[out=315,in=100,->] (13,1.5);
\draw[green] (15.3,3.7) edge[out=315,in=100,->] (16,1.5);

\draw[dotted] (1.7,1.5) edge[out=60,in=225,->] (5.7,3.6);
\draw[dotted] (4.7,1.5) edge[out=60,in=215,->] (9.8,3.6);
\draw[dotted] (7.7,1.5) edge[out=60,in=225,->] (13.7,3.6);

\draw[dotted] (9.7,1.5) edge[out=100,in=315,->] (6.7,3.5);
\draw[dotted] (10.3,1.5) edge[out=60,in=345,->] (10.7,3.5);

\draw[dotted] (12.7,1.5) edge[out=120,in=345,->] (7.5,3.6);
\draw[dotted] (13.3,1.5) edge[out=60,in=325,->] (14.7,3.5);

\draw[dotted] (15.7,1.5) edge[out=100,in=325,->] (11.5,3.7);
\draw[dotted] (16.2,1.5) edge[out=60,in=325,->] (15.5,3.8);
\end{tikzpicture}
\end{center}
\end{figure}

Note that we are only encoding the repair process.  We say nothing about $m$, 
$q$, $r$, $k$, $\beta$ and $\alpha$. If a bipartite digraph satisfies
the repair condition  it still doesn't say if it can be realised by
any parameters.  We call the digraph $\graph$ realisable  if
there is $(m, q, r, k, \beta, \alpha)$ such that there is an 
$(m; n, k, r, \alpha, \beta)$-functional repair code (FRC) $({\rm PG}(m-1,q), \A)$
with $\graph(\A) \equiv \graph$.

\section{Further work}
\label{sec:further}

The construction of Theorem \ref{thm:lines} does not require the
projective plane to be Desarguesian.  This naturally leads to the
question of whether one could construct more functional repair codes
from designs, since linearity is not required.  This approach may be
useful for functional repair code requiring repair-by-transfer
(\cite{ShumHu, Babu-Kumar, SilbersteinEtzion}), where the nodes
contributing information for repair do not perform any computations.
There has also been studies of locally repairable codes via matroid
theory (\cite{AlmostAffine,Matroid}) which may also be of interest for
functional repair codes.

Construction \ref{cons:dualarc} gives a functional repair code that is
flexible in terms of locality and availability for node repairs. There
are some recent work (\cite{local-avail}) in \emph{symbol localilty
  and availability}: not necessarily repairing whole nodes but only
some symbols in a node.  It would be intresting to see how this translate
into projective geometry.
 
The focal spread construction in Section \ref{sec:strictly} gives the
only known example of a strictly functional repair code.  However, it
is not clear whether a generalisation to larger fields or to higher
dimensions would retain this property.  Indeed, it is not even clear
whether one could still have a succinct description of the repair
process.  This indicates that there is still much to understand about
this interesting structure.  It is also not clear whether the
distilling of the properties of functional repair from this focal
spread construction into a non-recursive definition (Section
\ref{sec:simpler}) may be generalised. Again, this indicates that
further study of this structure may be profitable.

The view of a functional repair code as a digraph allows some
characterisation of exact repair codes.  However, as yet it is not
clear when a digraph with the right properties are actually realisable
as a functional repair code.  Another aspect to consider is: given a
digraph, is it always possible to ``complete'' it so
that it satisfies the repair condition or are there cases where this is 
impossible?

\end{document}